\documentclass[conference]{IEEEtran}
\IEEEoverridecommandlockouts
\ifCLASSINFOpdf
\else
\fi
\usepackage[many]{tcolorbox}

\usepackage{amsfonts}
\usepackage{bm}
\usepackage{amsthm}
\usepackage{amssymb}
\usepackage{mathabx}
\usepackage{changepage}
\usepackage{amsmath,bm}
\usepackage{mathtools}
\usepackage{pgfplots}

\usepackage{graphics}
\usepackage{graphicx}
\usepackage{pifont}
\usepackage{subfigure}
\usepackage{arydshln}
\usepackage{tikz}
\usepackage{textcomp}
\usepackage{empheq}
\usepackage{romannum}
\usepackage{algorithm}
\usepackage{algpseudocode}
\usetikzlibrary{tikzmark,calc,decorations.pathreplacing}
\usetikzlibrary{decorations,arrows,shapes}
\usetikzlibrary{arrows.meta}
\usepackage{array}
\usetikzlibrary{spy}

\usepackage{scalerel}
\usepackage{enumerate}
\usepackage{color}
\usepackage{hyperref}
\usepackage{algorithm}
\usepackage[english]{babel}

\newtheorem{theorem}{Theorem}

\theoremstyle{remark}
\newtheorem*{remark}{Remark}

\usepackage[capitalize]{cleveref}
\newcommand{\mf}{\mathbf}

\newcommand{\bl}[1]{\textcolor[rgb]{0,0,0}{#1}}
\newcommand{\rd}[1]{\textcolor[rgb]{0,0,0}{#1}}
\newcommand{\pp}[1]{\textcolor[rgb]{0,0,0}{#1}}
\newcommand{\ff}[1]{\textcolor[rgb]{0,0,0}{#1}}

\begin{document}
%
\title{Harnessing Holes for Spatial Smoothing with Applications in Automotive Radar}

\author{\IEEEauthorblockN{Yinyan Bu, Robin Rajam\"{a}ki, Pulak Sarangi, and Piya Pal}\\
\IEEEauthorblockA{Department of Electrical and Computer Engineering, University of California, San Diego}
}

\maketitle

\IEEEpeerreviewmaketitle
\begin{abstract}
    \bl{This paper studies spatial smoothing using sparse arrays in single-snapshot Direction of Arrival (DOA) estimation. We consider the application of automotive MIMO radar, which traditionally synthesizes a large uniform virtual array by appropriate waveform and physical array design. }
   We \bl{explore} deliberately introducing holes into this virtual array to leverage resolution gains provided by the increased aperture. The presence of these holes requires re-thinking DOA estimation, as conventional algorithms may no longer be easily applicable \bl{and alternative techniques, such as array interpolation, may be computationally expensive}. 
   Consequently, we study sparse array geometries that permit the direct application of spatial smoothing. 
   \bl{We show that a sparse array geometry is amenable to spatial smoothing if it can be decomposed into the sum set of two subsets of suitable cardinality. Furthermore, we demonstrate that many such decompositions may exist---not all of them yielding equal identifiability or aperture. } 
   We derive necessary and sufficient conditions to guarantee identifiability \bl{of a given number of targets}
   , which gives insight into choosing desirable decompositions \bl{for spatial smoothing}. This \bl{provides} uniform recovery guarantees and enables estimating DOAs \bl{at increased resolution and reduced} computational complexity.\footnote{This work was supported in part by \pp{Texas Instruments and} grants ONR N00014-19-1-2256, NSF 2124929, and DE-SC0022165.} 
\end{abstract}

\section{Introduction}
Active sensing provides several advantages compared to passive sensing by virtue of the flexibility of designing the transmitted waveforms for probing the environment. A prominent example of an active sensing system is multiple-input multiple-output (MIMO) radar, which is capable of achieving high angular resolution using only a limited number of physical sensors. By combining transmission of orthogonal waveforms and design of structured sparse array geometries, MIMO radar can synthesize a virtual array (also known as the sum co-array) with $M_t M_r$ elements using $M_t$ transmit and $M_r$ receive antennas. Sparse array geometries, such as the minimum-redundancy array \cite{moffet1968minimum} and nested array \cite{pal2010nested}, yield a large virtual uniform linear array (ULA) which naturally leads to enhanced resolution. Hence, sparse array-based MIMO radar has become extremely lucrative for deployment in domains with high resolution requirements, such as advanced driver assistant systems (ADAS) and autonomous vehicles. However, a key challenge imposed by automotive applications is that the environment is highly dynamic with rich multipath \cite{patole2017automotive,sun20214d}. As a result, the number of snapshots in a given coherence interval is small; in the worst case, only a single snapshot is available. This raises a fundamental question ``How can super-resolution direction of arrival (DOA) estimation be achieved using a single snapshot?"

A common practice in the MIMO radar literature \cite{li2007mimo} is to synthesize a virtual array with a large contiguous ULA segment of up to $O(M^2)$ virtual sensors, when $M_t \propto M_r \propto M$. The inherent ``shift invariant" structure of the (virtual) ULA can be leveraged to identify the ``signal subspace'' of interest from the single-snapshot measurement model. In particular, spatial smoothing \cite{odendaal1994two,pillai1989forward} can be used to accumulate multiple partial measurement vectors (corresponding to appropriate sub-arrays of the ULA) to build a spatially smoothed measurement matrix that is no longer rank-deficient. High-resolution DOA estimation can then be achieved by applying subspace-based methods, such as MUSIC \cite{schmidt1986multiple} and ESPRIT \cite{roy1989esprit}, on the spatially smoothed measurement matrix. The aperture of the virtual array can be further extended by making it a sparse array, which inevitably leads to the introduction of holes (given the same sensor budget). This requires re-thinking the application of conventional algorithms \bl{leveraging the structure of the ULA}. One way to tackle these holes is the recently proposed interpolation techniques based on low-rank Toeplitz or Hankel matrix completion \pp{\cite{sun2020sparse,sun20214d,sarangi2022single}} to obtain an interpolated virtual ULA. 
Upon successful interpolation, standard \bl{ULA-based} spatial smoothing \bl{is} applicable. However, there are two major challenges associated with \pp{such} interpolation techniques. Firstly, obtaining theoretical guarantees even in absence of noise for successful virtual array interpolation is difficult for arbitrary array geometries; secondly, the computational cost of interpolation via rank-minimization (or corresponding convex relaxations) can be \bl{prohibitively} high. 
A natural question is \pp{therefore}: \bl{How can interpolation-free methods, such as spatial smoothing, be applied on sparse (virtual) arrays while guaranteeing the identifiability of a desired number of targets?}

\textbf{Contributions:}
 This paper explores 
 synthesizing (virtual) \bl{sparse} arrays 
 to enhance resolution 
 \bl{compared} to conventional uniform (virtual) arrays. 
 \bl{We characterize the set of sparse array geometries amenable to spatial smoothing, 
 \pp{establishing} that several decompositions may exist for a given array geometry and spatial smoothing parameter values, but some decompositions may be preferable to others.} We \bl{derive} necessary and sufficient conditions for \bl{uniquely} identifying $K$ targets (in absence of noise) 
 \bl{using} spatial smoothing on \bl{these} sparse arrays. 
 \bl{We demonstrate that leveraging holes in spatial smoothing using appropriately designed sparse (virtual) arrays }
 can \bl{improve resolution without introducing ambiguities.}
 
\textbf{Notation:} Given an array geometry $\mathbb{S}=\{d_1,d_2,\cdots,d_N\}$, matrix $\mathbf{A}_{\mathbb{S}}(\bm{\theta}) \in \mathbb{C}^{N\times K}$ denotes the array manifold for sensors located at $n\lambda/2$, where $n\in \mathbb{S}$. \bl{The $(n,k)$th entry of $\mathbf{A}_{\mathbb{S}}(\bm{\theta})$ is} $[\mathbf{A}_{\mathbb{S}}(\bm{\theta})]_{n,k}=\exp(j\pi d_n \sin\theta_k)$, where $\bm{\theta}\in[-\pi/2,\pi/2)^K$ denote the target DOAs. 
We use the first sensor as the reference sensor ($d_1=0$). 
The Khatri-Rao \bl{(column-wise Kronecker)} product is denoted by $\odot$. Moreover, $\mathcal{R}(\cdot)$ and $\mathcal{N}(\cdot)$ denote the range (column space) and null space, respectively.

\section{Measurement Model}
Consider $K$ narrowband sources impinging on \pp{a} linear array $\mathbb{S}=\{d_1,d_2,\cdots,d_N\}$ from distinct angular directions $\bm{\theta}=[\theta_1, \theta_2, \cdots, \theta_K]^T$. \pp{A single temporal} snapshot \pp{of the} received signal is of the following form: 
\begin{equation}\label{eqn:full_measurement}
\mathbf{y}=\bm{A}_{\mathbb{S}}(\bm{\theta})\bm{x} + \mathbf{n},
\end{equation}
where \bl{$\bm{x}\in\mathbb{C}^K$} is the source/target signal and \bl{$\mathbf{n}\in\mathbb{C}^{N}$} is a noise vector. The goal is to estimate $\{\theta_k\}_{k=1}^K$ \bl{given} $\mathbf{y}$ \bl{and $\mathbb{S}$}.

Note that \eqref{eqn:full_measurement} is applicable to both passive and active sensing---indeed, $\mathbb{S}$ can represent either a physical or virtual array. In case of co-located MIMO radar using orthogonal waveforms, \bl{which is the focus of this paper,} $\mathbb{S}=\mathbb{S}_t+\mathbb{S}_r$ is a sum co-array, where $\mathbb{S}_t$ and $\mathbb{S}_r$ are the transmitter and receiver arrays, \bl{respectively}. For ease of exposition, we restrict ourselves to non-redundant arrays, resulting in $\vert \mathbb{S}\vert=M_tM_r$, \bl{where $M_t\triangleq |\mathbb{S}_t|$ and $M_r\triangleq |\mathbb{S}_r|$}. 

\section{Spatial Smoothing with Sparse Arrays}

Spatial-smoothing using ULAs has been widely used for DOA estimation in sample-starved regimes or to tackle coherent sources \cite{shan1985spatial,pillai1989forward}. \pp{While it has been recognized that non-uniform arrays with a suitable shift-invariant structure can be used for spatial smoothing \cite{friedlander1992direction,wang19982} (or directly in algorithms such as ESPRIT \cite{xu2023coprime}),
the principled design of such ``spatial-smoothing-amenable'' sparse linear arrays geometries providing enhanced aperture with rigorous identifiability guarantees has not yet been fully explored.} 
\pp{In this section, we study spatial smoothing using sparse subarrays, developing new results on source/target identification which show that holes can be introduced in the array to increase aperture without compromising identifiability.}

An array $\mathbb{S}$ is said to have a shift-invariant structure if $\mathbb{S}$ contains shifted copies of a \bl{so-called ``basic sub-array''} $\mathbb{S}_b\subset\mathbb{S}$. Mathematically: $\bigcup_{i=1}^{L} (\mathbb{S}_b+\delta_i) \subseteq \mathbb{S}$, where $\delta_i\rd{\in\mathbb{Z}}$ denotes the $i$-th \rd{(unique integer-valued)} shift, $i\in\{1,2,\cdots,L\}$. 
We define the set of all $N$-sensor linear arrays that are amenable for (forward) spatial smoothing with parameters $N_s$ and $L$ as:
\begin{equation}
\begin{aligned}    
\mathcal{S}_{N}(N_s,L)\triangleq
\{&\pp{\mathbb{S},\vert\mathbb{S}\vert=N \text{ such that } \exists\ \mathbb{S}_b,\mathbb{S}_c \text{ with}}\\ &\pp{\vert\mathbb{S}_{b}\vert=N_s,\lvert\mathbb{S}_c\rvert=L,\mathbb{S}_b+\mathbb{S}_c \subseteq \mathbb{S}  \}.}
\end{aligned}\label{def:set}
\end{equation}
\pp{In other words,} set $\mathcal{S}_{N}(N_s,L)$ represents all linear arrays $\mathbb{S}$ that constitute of $L$ sub-arrays $\{\mathbb{S}_i\}_{i=1}^L=\{\mathbb{S}_b+\delta_i\}_{i=1}^L$, which may or may not be overlapping, that are integer-shifted copies of a basic sub-array $\mathbb{S}_b$ with $N_s$ sensors. Note that the same array $\mathbb{S}$ may belong to both $\mathcal{S}_{N}(N_1,L_1)$ and $\mathcal{S}_{N}(N_2,L_2)$ with $N_1\neq N_2$ or $L_1 \neq L_2$, i.e., there could exist multiple decompositions for the same $\mathbb{S}$. However, not all of them are equivalent in terms of identifiability and resolution, as we will show in \bl{\cref{sec:decomposition_considerations}}. 

If $\mathbb{S}\in \mathcal{S}_{N}(N_s,L)$, we can construct a spatially smoothed measurement matrix $\mathbf{Y}$ by rearranging measurement vector $\mathbf{y}$ into an $N_s\times L$ matrix as follows:
\begin{equation}
    \mathbf{Y}=[\bm{y}_1,\bm{y}_2,\cdots,\bm{y}_L],
\end{equation}
where $\bm{y}_i$ contains the elements of $\mathbf{y}$ corresponding to subarray $\mathbb{S}_i$. Specifically, in the absence of noise ($\mathbf{n}=\mathbf{0}$), $\bm{y}_i=\bm{A}_{\mathbb{S}_i}(\bm{\theta})\bm{x}=\bm{A}_{\mathbb{S}_b}(\bm{\theta})\bm{D}_i(\bm{\theta})\bm{x}$, 
where $\bm{D}_i(\bm{\theta})$ is a diagonal matrix with $[\bm{D}_i(\bm{\theta})]_{m,m}=\exp(j\pi\delta_i\sin\theta_m)$. Due to the shift-invariance property of the array, the spatially smoothed measurement matrix permits the following decomposition (when $\mathbf{n}=\mathbf{0}$):
\begin{align}
\mathbf{Y}&=\bm{A}_{\mathbb{S}_b}(\bm{\theta})[\bm{D}_1(\bm{\theta})\bm{x},\bm{D}_2(\bm{\theta})\bm{x},\cdots,\bm{D}_{L}(\bm{\theta})\bm{x}]\nonumber\\
&=\bm{A}_{\mathbb{S}_b}(\bm{\theta})\text{diag}(\bm{x})\bm{A}_{\mathbb{S}_c}(\bm{\theta})^{T}. \label{eqn:ss_op}
\end{align}
\pp{In the} presence of noise, $\mathbf{Y}$ has an additive term, where noise vector $\mathbf{n}$ is reshaped \pp{according to the shift structure.}

\cref{eqn:ss_op} illustrates that the shift-invariant structure of the array 
$\mathbb{S}$ can \bl{be leveraged} to \bl{potentially} build the rank of $\mathbf{Y}$ on which subspace methods can be applied to identify $\bm{\theta}$. The following \lcnamecref{theorem:1} provides necessary and sufficient conditions for identifying the desired subspace corresponding to the true DOAs $\bm{\theta}$ by applying MUSIC on $\mf{Y}$ \pp{using a sparse subarray $\mathbb{S}_b$ with holes}. \bl{When these conditions are satisfied, the} MUSIC pseudo-spectrum yields exactly $K$ peaks, and no false peaks.
\begin{theorem}\label{theorem:1}
Consider the measurement model \eqref{eqn:full_measurement} with $\mathbf{n}=\mathbf{0}$, suppose $\mathbb{S}\in \mathcal{S}_{N}(N_s,L)$. Applying MUSIC on $\mathbf{Y}$ in \eqref{eqn:ss_op} can resolve any $K < \pp{\min(N_s,L+1)}$ distinct angles $\{\theta_k\}_{k=1}^K$ unambiguously if and only if both of the following conditions hold: \\
(a) $\bm{A}_{\mathbb{S}_b}(\bm{\phi})\in \mathbb{C}^{N_s\times (K+1)}$ satisfies rank$(\bm{A}_{\mathbb{S}_b}(\bm{\phi}))=K+1$ for all possible sets of $K + 1$ distinct $\phi_i$ in $[-\frac{\pi}{2},\frac{\pi}{2})$; \\
(b) $\bm{A}_{\mathbb{S}_c}(\bm{\vartheta})\in\mathbb{C}^{L\times K}$ satisfies $rank(\bm{A}_{\mathbb{S}_c}(\bm{\vartheta}))=K$ for all possible sets of $K$ distinct $\vartheta_i$ in $[-\frac{\pi}{2},\frac{\pi}{2})$.
\end{theorem}
\begin{proof}
    Let the singular value decomposition of $\mathbf{Y}$ be $\mathbf{Y}=\mathbf{U\Sigma V}^H$ \rd{where $\mathbf{\Sigma}$ is a diagonal matrix containing the singular values of $\mathbf{Y}$ in descending order.} The singular vectors are partitioned \rd{according to the number of non-zero singular values (also equal to rank($\mathbf{Y}$)) $\hat{K}\leq K$} as $\mathbf{U}=[\mathbf{U}_s,\mathbf{U}_n]$, where $\mathbf{U}_s\in\mathbb{C}^{N_s\times \rd{\hat{K}}}$ and $\mathbf{U}_n\in\mathbb{C}^{N_s\times(N_s-\rd{\hat{K}})}$. \rd{MUSIC applied on $\mathbf{Y}$ is said to unambiguously identify any set of $K< \pp{\min(N_s,L+1)}$ distinct sources $\{\theta_k\}_{k=1}^{K}$ if} the following two conditions are satisfied: (M1) $\mathbf{a}_{\mathbb{S}_b}(\theta_k)^H\mathbf{U}_n\mathbf{U}_n^{H}\mathbf{a}_{\mathbb{S}_b}(\theta_k)=0$ for all $1\leq k \leq K$, and 
    (M2) $\mathbf{a}_{\mathbb{S}_b}(\phi)^H\mathbf{U}_n\mathbf{U}_n^{H}\mathbf{a}_{\mathbb{S}_b}(\phi)\neq 0$
    for any $\phi \not \in \{\theta_k\}_{k=1}^{K}$.

    We begin by proving the sufficiency of ($a$) and ($b$) for identifiability. Suppose ($a$) and ($b$) hold. From \eqref{eqn:ss_op} we have $\mathcal{R}(\mathbf{Y})=\mathcal{R}(\bm{A}_{\mathbb{S}_b}(\bm{\theta})\text{diag}(\bm{x})\bm{A}_{\mathbb{S}_c}^{T}(\bm{\theta}))$. Due to assumption ($b$), $\text{diag}(\bm{x})\bm{A}_{\mathbb{S}_c}^{T}(\bm{\theta})$ has full row-rank and hence $\mathcal{R}(\mathbf{Y})=\mathcal{R}(\bm{A}_{\mathbb{S}_b}(\bm{\theta}))$. 
    %
    \pp{Due to assumption (a) rank($\bm{A}_{\mathbb{S}_b}(\bm{\theta}))=K$ and hence $\hat{K}=K$}. Now, for every $\{\theta_k\}_{k=1}^{K}$, we have $\mathbf{a}_{\rd{\mathbb{S}_b}}(\theta_k)\in \mathcal{R}(\mathbf{U}_s)$ and therefore $\mathbf{a}_{\rd{\mathbb{S}_b}}(\theta_k)^H\mathbf{U}_n\mathbf{U}_n^{H}\mathbf{a}_{\rd{\mathbb{S}_b}}(\theta_k)=0$. \pp{The fact} \rd{that $\{\theta_k\}_{k=1}^K$ are the only solutions follows by contradiction: } Suppose there exists $\phi \not\in \{\theta_k\}_{k=1}^K$ \rd{such that} $\mathbf{a}_{\rd{\mathbb{S}_b}}(\phi)^H\mathbf{U}_n\mathbf{U}_n^{H}\mathbf{a}_{\rd{\mathbb{S}_b}}(\phi)=0$. This implies that $\mathbf{U}_n^{H}\mathbf{a}_{\rd{\mathbb{S}_b}}(\phi)=0 \Rightarrow \mathbf{a}_{\rd{\mathbb{S}_b}}(\phi)\in \mathcal{R}(\mathbf{U}_s)=\mathcal{R}(\mathbf{A}_{\mathbb{S}_b}(\bm{\theta}))$. However, this leads to a contradiction since $[\mathbf{A}_{\mathbb{S}_b}(\bm{\theta}),\mathbf{a}_{\rd{\mathbb{S}_b}}(\phi)]$ cannot be rank deficient due to assumption ($a$). 
    Thus, if ($a$) and ($b$) hold, MUSIC applied on $\mathbf{Y}$ can identify any set of $K$ distinct sources unambiguously.

     Next we prove the necessity of ($a$) and ($b$) for identifiability. 
%
    \pp{We consider an arbitrary set of $K$ source angles $\{\theta_k\}_{k=1}^{K}$.} 
    \pp{We always have} $\mathcal{R}(\mathbf{Y})\subseteq \mathcal{R}(\mathbf{A}_{\mathbb{S}_b}(\bm{\theta}))$. However, we establish that if \rd{conditions} (M1) and (M2) hold, then it is necessary to have  
    $\mathcal{R}(\mathbf{Y})= \mathcal{R}(\mathbf{A}_{\mathbb{S}_b}(\bm{\theta}))$ and \pp{$\hat{K}=K$}. The proof proceeds via contradiction. Suppose $\mathcal{R}(\mathbf{Y})\subset \mathcal{R}(\mathbf{A}_{\mathbb{S}_b}(\bm{\theta}))$, i.e., there exists some source direction $\theta_k$ such that $\pp{\mathbf{a}_{\mathbb{S}_b}(\theta_k) \not\in \mathcal{R}(\mathbf{Y})=\mathcal{R}(\mathbf{U}_s)}$. This implies that the projection of $\mathbf{a}_{\rd{\mathbb{S}_b}}(\theta_k)$ onto $\mathcal{R}(\mathbf{U}_n)$---the orthogonal complement of $\mathcal{R}(\mathbf{Y})$---is non-zero. Hence, we have $\mathbf{a}_{\rd{\mathbb{S}_b}}(\theta_k)^H\mathbf{U}_n\mathbf{U}_n^H\mathbf{a}_{\rd{\mathbb{S}_b}}(\theta_k)\neq 0$ which contradicts (M1). Thus, $\mathcal{R}(\mathbf{Y})= \mathcal{R}(\mathbf{A}_{\mathbb{S}_b}(\bm{\theta}))$.

    Next, we establish that for unambiguous identification of any set of $K<\pp{\min(N_s,L+1)}$ sources it is necessary that rank$(\mathbf{Y})=K$. Suppose rank$(\mathbf{Y})=K'<K$. \rd{Since $\mathcal{R}(\mathbf{Y})= \mathcal{R}(\mathbf{A}_{\mathbb{S}_b}(\bm{\theta}))$,} rank$(\mathbf{A}_{\mathbb{S}_b}(\bm{\theta}))=\text{rank}(\mathbf{Y})=K'<K$. This implies that there exists $\{\bar{\theta}_k\}_{k=1}^{K'}\subset \{\theta_k\}_{k=1}^{K}$ \pp{such that $\{\mathbf{a}_{\mathbb{S}_b}(\bar{\theta}_k)\}_{k=1}^{K'}$ are linearly independent}. 
    \pp{Now consider a different source configuration where the DOAs are given by
    $\bm{\bar{\theta}}=[\bar{\theta}_1,\bar{\theta}_2,\ldots,\bar{\theta}_{K'}]^T$ and the corresponding measurement vector is $\bm{\bar{y}}=\mathbf{A}_{\mathbb{S}}(\bm{\bar{\theta}})\bm{\bar{x}}$ for some $\bm{\bar{x}}\in \mathbb{C}^{K'}$ (with $\bar{x}_i\neq 0$ for all $1\leq i \leq K'$)
    . The corresponding spatially smoothed matrix is given below:}
    \begin{align*}    \pp{\mathbf{\bar{Y}}=\mathbf{A}_{\mathbb{S}_b}(\bm{\bar{\theta}})\text{diag}(\bm{\bar{x}})\mathbf{A}_{\mathbb{S}_c}(\bm{\bar{\theta}})^{T}}.
    \end{align*}
    Identifiability of all $K<\pp{\min(N_s,L+1)}$ source configurations implies that these $K'<K<\pp{\min(N_s,L+1)}$ DOAs $\bm{\bar{\theta}}$ should also be identifiable by applying MUSIC on $\mathbf{\bar{Y}}$.
    \pp{However, since $\mathcal{R}(\mathbf{Y})=\mathcal{R}(\mathbf{\bar{Y}})=\mathcal{R}(\mathbf{A}_{\mathbb{S}_b}(\bm{\theta}))=\mathcal{R}(\mathbf{A}_{\mathbb{S}_b}(\bm{\bar{\theta}}))$ MUSIC algorithm applied to $\mathbf{\bar{Y}}$ will produce $K>K'$ peaks and $\bar{\theta}_1,\bar{\theta}_2,\cdots,\bar{\theta}_{K'}$ will not be identifiable.}
    This contradicts the fact that all $K<\pp{\min(N_s,L+1)}$ source configurations can be unambiguously \rd{identified by applying MUSIC on $\bm{Y}$.} Therefore, rank$(\mathbf{Y})=\text{rank}(\mathbf{A}_{\mathbb{S}_b}(\bm{\theta}))=K$ is necessary. \pp{This implies condition ($b$) since rank$(\mathbf{Y})\leq \text{rank}(\mathbf{A}_{\mathbb{S}_c}(\bm{\theta}))\leq K$.}
    By combining $\text{rank}(\mathbf{A}_{\mathbb{S}_b}(\bm{\theta}))=K$ and condition (M2), we can establish that for any $\phi\not\in \{\theta_k\}_{k=1}^K$ we must have $\mathbf{a}_{\mathbb{S}_b}(\phi)\not\in \mathcal{R}(\mathbf{A}_{\mathbb{S}_b}(\bm{\theta}))$ establishing the necessity of ($a$).  
\end{proof}

\begin{remark}
\cref{theorem:1} \bl{reveals} two potential ambiguities in spatial smoothing using sparse arrays. First, when ($a$) does not hold, we have rank$(\bm{A}_{\mathbb{S}_b}(\bm{\theta}))\!<\!K$ or there exists $\varphi\!\notin\!\{\theta_i\}_{i=1}^K$ such that \pp{rank$([\bm{A}_{\mathbb{S}_b}(\bm{\theta}), \bm{a}_{\mathbb{S}_b}(\varphi)])\!\leq \!K$}, where $\bm{a}_{\mathbb{S}_b}(\varphi)$ is the manifold vector of $\mathbb{S}_b$ for angle $\varphi$. Second, when ($b$) does not hold, we have 
$\mathcal{R}(\mathbf{Y})\neq\mathcal{R}(\bm{A}_{\mathbb{S}_b}(\bm{\theta}))$. 
In both cases, MUSIC \pp{fails} to yield exactly $K$ peaks that corresponds to the true $K$ DOAs.
\end{remark}
\subsection{\bl{Non-necessity of ULA segments}}

\bl{Using \cref{theorem:1}, it can be shown that $\mathbb{S}_b$ with a ULA segment of length $K+1$ and $\mathbb{S}_c$ with a ULA segment of length $K$ are sufficient for guaranteeing identifiability of $K$ targets \pp{due to the presence of a Vandermonde submatrix in $\mathbf{A}_{\mathbb{S}_b}(\bm{\theta})$ and $\mathbf{A}_{\mathbb{S}_c}(\bm{\theta})$}. Hence, choosing $\mathbb{S}_b$ and $\mathbb{S}_c$ as sparse arrays (such as nested arrays \cite{sarangi2022single}) containing contiguous segments of size $K+1$ and $K$ is sufficient for satisfying conditions ($a$) and ($b$), respectively. However, \pp{the presence of an ULA segment is not always necessary}: it has been shown in \cite{chen2021rank} that an $N'$-sensor array $\mathbb{D}$ \emph{need not} contain a ULA segment of \pp{length $K'+1$ to identify 
$K'\ff{<N'}$ 
sources when $N'=3$ or $4$}.} 

\rd{These results can be leveraged in \cref{theorem:1} to guarantee identifiability of $K=2$ targets when applying MUSIC on the spatially smoothed matrix $\bm{Y}$ using sparse $\mathbb{S}_b$ and $\mathbb{S}_c$---\pp{neither containing ULA segments of appropriate length $K$}}. 
{Specifically, by \cite[Theorem~2]{chen2021rank}, manifold matrix $\bm{A}_{\mathbb{D}}(\bm{\theta})\in\mathbb{C}^{3\times 2}$ of array $\mathbb{D}=\{0,d_1,d_2\}\subset\mathbb{N}$ has full column rank for all distinct $\theta_1,\theta_2\in[-\frac{\pi}{2},\frac{\pi}{2})$ if and only if $d_1,d_2$ are coprime. Hence, condition ($b$) in \cref{theorem:1} is satisfied for a given $L\geq 3$ by selecting $d_1,d_2$ to be coprime and setting $\mathbb{S}_c=\mathbb{D}\cup \mathbb{X}_1$ for any $\mathbb{X}_1$ such that $|\mathbb{D}\cup \mathbb{X}_1|=L$. This follows from the fact that $\bm{A}_{\mathbb{S}_c}(\bm{\theta})\in\mathbb{C}^{L\times 2}$ contains a $3\times 2$ block $\bm{A}_{\mathbb{D}}(\bm{\theta})$ which by \cite[Theorem~2]{chen2021rank} has full Kruskal rank for the above choice of $\mathbb{D}$. A similar, albeit slightly more invloved argument can be made for satisfying condition ($a$) in \cref{theorem:1} by embedding manifold matrix $\bm{A}_{\mathbb{G}}(\theta)\in\mathbb{C}^{4\times 2}$ of array $\mathbb{G}=\{0,g_1,g_2,g_3\}$ into $\bm{A}_{\mathbb{S}_{\ff{b}}}(\bm{\theta})\in\mathbb{C}^{N_s\times 2}$ for $g_1,g_2,g_3$ satisfying certain coprimality properties---see \cite[Theorem~3]{chen2021rank} for details. \pp{For example, \cref{theorem:1} combined with \cite[Theorems 2 and 3]{chen2021rank} suggests that spatial smoothing identifies $K=2$ targets using $\mathbb{S}_c=\{0,4,9\}$ and $\mathbb{S}_b=\{0,3,5,7\}$---neither containing ULA segments with unit spacing.}

\subsection{\pp{Harnessing holes in} sparse arrays for spatial smoothing: Decomposition considerations}\label{sec:decomposition_considerations}
It is well known that a length $2M-1$ ULA can resolve $M$ targets unambiguously via forward spatial smoothing in absence of noise by choosing $\mathbb{S}_b$ and $\mathbb{S}_c$ to be length-$M$ ULAs. However, trading off identifiability for resolution may be of interest for example in automotive radar, where the number of targets in a given range-Doppler bin can be small, yet high angular resolution is required \cite{sun2020mimo}. A critical factor that affects the spatial resolution of such a spatial-smoothing strategy (in addition to SNR) is the aperture of $\mathbb{S}_b$, denoted by $N_b:=\max(\mathbb{S}_b)-\min(\mathbb{S}_b)$. For a given number of virtual elements $N$, if we desire a larger $N_b$, the number of effective ``snapshots" $L$ is limited and vice-versa.
\newcommand{\figw}{9}
\newcommand{\fe}{1}
\newcommand{\figh}{2.0}
\newcommand{\ef}{2}
\newcommand{\Lmax}{50} 
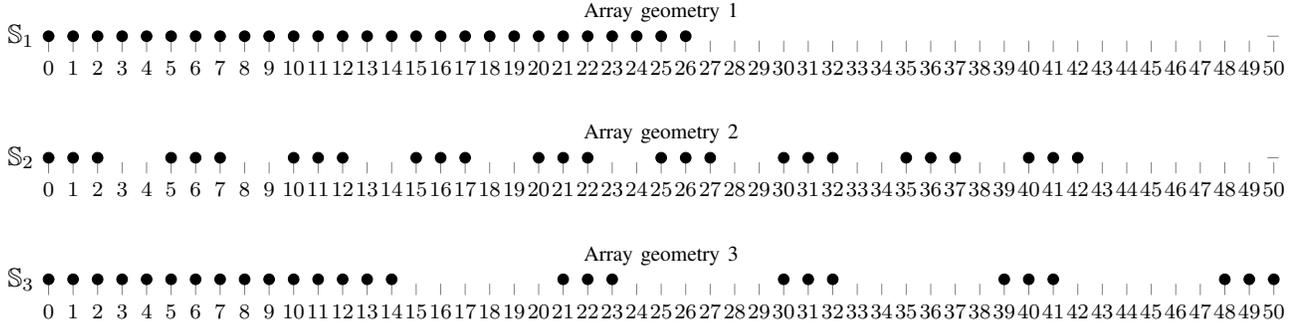
\begin{figure*}[ht]
\begin{tikzpicture}
\begin{axis}
[width=\figw*(1 cm+(\ef cm-1 cm)*\fe) ,height=\figh cm,xmin=-0.2,xmax=\Lmax+0.2,xtick={0,1,...,\Lmax},ytick={0},yticklabels={\pp{$\mathbb{S}_1$}},title style={yshift=0 pt},xticklabel shift = 0 pt,xtick pos=bottom,axis line style={draw=none},title={\footnotesize Array geometry 1},title style={yshift=-10 pt}, xticklabel style={font=\footnotesize}]
\addplot[only marks]
coordinates{(0,0)(1,0)(2,0)(3,0)(4,0)(5,0)(6,0)(7,0)(8,0)(9,0)(10,0)(11,0)(12,0)(13,0)(14,0)(15,0)(16,0)(17,0)(18,0)(19,0)(20,0)(21,0)(22,0)(23,0)(24,0)(25,0)(26,0)};
\end{axis}
\end{tikzpicture}
\\ \\
\begin{tikzpicture}
\begin{axis}
[width=\figw*(1 cm+(\ef cm-1 cm)*\fe) ,height=\figh cm,xmin=-0.2,xmax=\Lmax+0.2,xtick={0,1,...,\Lmax},ytick={0},yticklabels={\pp{$\mathbb{S}_2$}},title style={yshift=0 pt},xticklabel shift = 0 pt,xtick pos=bottom,axis line style={draw=none},title={\footnotesize Array geometry 2},title style={yshift=-10 pt}, xticklabel style={font=\footnotesize}]
\addplot[only marks]
coordinates{(0,0)(1,0)(2,0)(5,0)(6,0)(7,0)(10,0)(11,0)(12,0)(15,0)(16,0)(17,0)(20,0)(21,0)(22,0)(25,0)(26,0)(27,0)(30,0)(31,0)(32,0)(35,0)(36,0)(37,0)(40,0)(41,0)(42,0)};
\end{axis}
\end{tikzpicture}
\\ \\
\begin{tikzpicture}
\begin{axis}
[width=\figw*(1 cm+(\ef cm-1 cm)*\fe) ,height=\figh cm,xmin=-0.2,xmax=\Lmax+0.2,xtick={0,1,2,...,\Lmax},ytick={0},yticklabels={\pp{$\mathbb{S}_3$}},title style={yshift=0 pt},xticklabel shift = 0 pt,xtick pos=bottom,axis line style={draw=none},title={\footnotesize Array geometry 3},title style={yshift=-10 pt}, xticklabel style={font=\footnotesize}]
\addplot[only marks]
coordinates{(0,0)(1,0)(2,0)(3,0)(4,0)(5,0)(6,0)(7,0)(8,0)(9,0)(10,0)(11,0)(12,0)(13,0)(14,0)(21,0)(22,0)(23,0)(30,0)(31,0)(32,0)(39,0)(40,0)(41,0)(48,0)(49,0)(50,0)};
\end{axis}
\end{tikzpicture}
\caption{Three examples of array geometry with $M_t=3$ antennas and $M_r=12$ antennas: achieve larger aperture by introducing holes}\label{fig:arrays}
\end{figure*}
This trade-off between aperture and identifiability continues to hold for spatial smoothing using sparse arrays. For example, given $\mathbb{S}=\{0,1,3,4,5,6,7,8\}$, we have at least the following decomposition for $N_s=4, L=2$:
\begin{align}
    &\mathbb{S}_{b_1}=\{0,3,5,7\}, \mathbb{S}_{c_1}=\{0,1\};\label{eq:Sb1}
\end{align}
\bl{and for $N_s=4, L=3$:}
\begin{align}
    &\mathbb{S}_{b_2}=\{0,1,3,4\}, \mathbb{S}_{c_2}=\{0,3,4\};\label{eq:Sb3}\\
    &\mathbb{S}_{b_3}=\{0,3,4,5\}, \mathbb{S}_{c_3}=\{0,1,3\}.\label{eq:Sb4}
\end{align} 
Using \cref{theorem:1} and \cite[Theorem 3]{chen2021rank}, 
both \labelcref{eq:Sb1,eq:Sb4} can guarantee the identifiability of $K=2$ targets. However, $\mathbb{S}_{b_1}=\{0,3,5,7\}$ has larger aperture than $\mathbb{S}_{b_3}=\{0,3,4,5\}$. Hence, $\mathbb{S}_{b_1}$ may be preferable to $\mathbb{S}_{b_3}$ in terms of resolution. Moreover, even though \eqref{eq:Sb3} \pp{may} not be desirable for spatial smoothing due to the fact $\mathbb{S}_{b_2}=\{0,1,3,4\}$ fails to satisfy condition ($a$) \bl{in \cref{theorem:1}} for $K=2$ according to \cite[Theorem 3]{chen2021rank}, $\mathbb{S}_{b_2}$ can be good candidate for the physical \bl{transmit or receive array if the aperture available for placing the sensors is tightly constrained.}

\section{Simulations}
For the numerical experiments, we consider three different array configurations with the constraint \pp{$M_t= 3$} and \pp{$M_r = 9$}. \cref{fig:arrays} visualizes the virtual sum co-array for each of these arrays, denoted by $\mathbb{S}_1, \mathbb{S}_2$, and $ \mathbb{S}_3$, respectively.

We generate the scattering coefficients $\{x_k\}_{k=1}^K$ with a constant magnitude and phase drawn uniformly from the interval $(0,2\pi]$. The noise is assumed to be complex Gaussian with independent real and imaginary parts and a variance $\sigma^2$ chosen to meet the desired signal-to-noise ratio (SNR) defined as $\text{SNR}=20\log_{10}\frac{\min_k(\lvert x_k\rvert)}{\sigma}$. \cref{fig:res} shows \bl{a realization of the MUSIC pseudospectrum in the case of two different target configurations (with $K=2$)} after performing spatial smoothing on \bl{$\mathbb{S}_i, i=1,2,3$}. The corresponding physical arrays \bl{$\mathbb{S}_{r_i},\mathbb{S}_{t_i}$} and \bl{spatial} smoothing \bl{arrays $\mathbb{S}_{b_i},\mathbb{S}_{c_i}$ (satisfying $\mathbb{S}_{r_i}+\mathbb{S}_{t_i}=\mathbb{S}_{b_i}+\mathbb{S}_{c_i}= \mathbb{S}_i$)} are shown below:
\begin{align*}
\mathbb{S}_1:\enspace&\mathbb{S}_{r_1}=\{3k\}_{k=0}^8,\thinspace\mathbb{S}_{t_1}=\{0,1,2\},\\
&\mathbb{S}_{b_1}=\{k\}_{k=0}^{24},\thinspace\mathbb{S}_{c_1}=\{0,1,2\};\\
\mathbb{S}_2:\enspace&\mathbb{S}_{r_2}=\{5k\}_{k=0}^8,\thinspace\mathbb{S}_{t_{\bl{2}}}=\{0,1,2\},\\
&\mathbb{S}_{b_2}=\bigcup_{n=0}^{6}\{k+5n\}_{k=0}^{2},\thinspace\mathbb{S}_{c_2}=\{0,5,10\};\\
\mathbb{S}_3:\enspace&\mathbb{S}_{r_{\bl{3}}}=\{3k\}_{k=0}^4\bigcup\{21+9k\}_{k=0}^{3},\thinspace\mathbb{S}_{t_{\bl{3}}}=\{0,1,2\},\\
&\mathbb{S}_{b_3}=\{k\}_{k=0}^{12}\bigcup\{21+9k\}_{k=0}^{3},\thinspace\mathbb{S}_{c_3}=\{0,1,2\}.
\end{align*}

The left column of \cref{fig:res} shows that $\mathbb{S}_1$ cannot resolve the DOAs with a separation of $2^{\circ}$ due to its limited aperture, whereas the enhanced aperture of \bl{sparse arrays} $\mathbb{S}_2$ and $\mathbb{S}_3$ allows them to resolve the closely spaced targets despite the presence of holes in \bl{these (virtual) array geometries}. \bl{However,} the right column of \cref{fig:res} shows a target configuration that corresponds to an ambiguity of $\mathbb{S}_2$, which leads to a degradation in its performance despite the increased angular separation ($\approx67^\circ$) between the targets. \bl{In contrast, $\mathbb{S}_3$ provably does not suffer from ambiguities when $K=2$ by \cref{theorem:1}.} 
\begin{figure}[h]
    \includegraphics[width=0.241\textwidth]{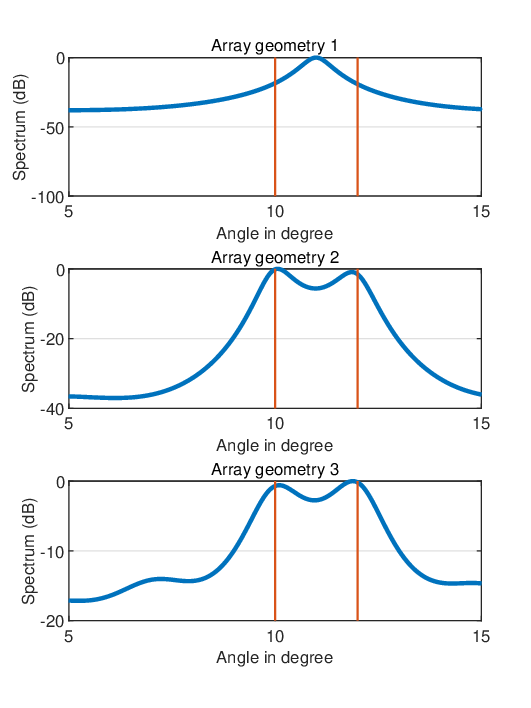}
    \includegraphics[width=0.241\textwidth]{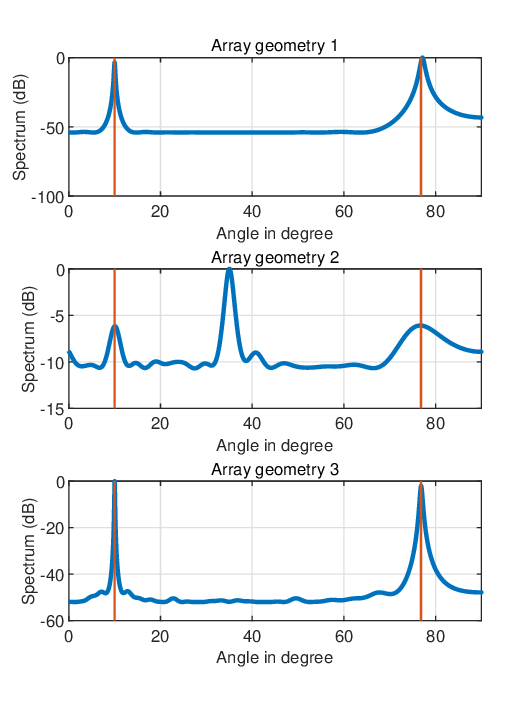}
    \caption{Comparison of direct MUSIC on 3 examples of array geometry with $K = 2$ \bl{targets} located at (left) $\bm{\theta}=[10^{\circ}, 12.00^{\circ}]$ and (right) $\bm{\theta}=[10^{\circ}, 76.82^{\circ}]$ with SNR $ = 20$ dB. \bl{Array} $\mathbb{S}_3$ can not only achieve higher resolution but is also free from ambiguities arising from spatial smoothing.} \label{fig:res}
\end{figure}
\section{Conclusions}
This paper considered leveraging the enhanced aperture of sparse arrays in spatial smoothing towards improving angular resolution in sample-starved applications, such as automotive radar. We showed that sparse arrays \bl{with a suitable shift-invariant structure are} amenable to spatial smoothing. 
Furthermore, we derived necessary and sufficient conditions for spatial smoothing using sparse arrays to identify $K$ targets. 
We also demonstrated that a fixed sparse array geometry can have multiple decompositions and not all of them are equivalent in terms of resolution and identifiability. Simulation results indicate that appropriate sparse array configuration can not only achieve higher resolution but can also be free of ambiguities when using spatial smoothing.

\section*{\pp{Acknowledgements}}
\pp{The authors would like to thank Mehmet Can H\"uc\"umeno\u{g}lu for valuable discussions and suggestions regarding \cref{theorem:1}.}

\bibliographystyle{IEEEtran}
\bibliography{ref}
\end{document}